\documentclass[conference]{IEEEtran}
\IEEEoverridecommandlockouts
\usepackage{url}
\usepackage{comment}
\usepackage{amsthm}
\usepackage{amssymb}
\usepackage{mathtools}
\usepackage{algorithm}
\usepackage{algorithmic}
\usepackage{subcaption}
\def\BibTeX{{\rm B\kern-.05em{\sc i\kern-.025em b}\kern-.08em
    T\kern-.1667em\lower.7ex\hbox{E}\kern-.125emX}}

\theoremstyle{definition}
\newtheorem{definition}{Definition}[section]
\newtheorem{lemma}{Lemma}[section]

\newcommand{\bh}{\beta}

\title{Enabling Bitcoin Smart Contracts\\on the Internet Computer}
\author{\IEEEauthorblockN{Ryan Croote}
\IEEEauthorblockA{\textit{DFINITY} \\
ryan.croote@dfinity.org}
\and
\IEEEauthorblockN{Islam El-Ashi}
\IEEEauthorblockA{\textit{DFINITY} \\
islam.elashi@dfinity.org}
\and
\IEEEauthorblockN{Thomas Locher}
\IEEEauthorblockA{\textit{DFINITY} \\
thomas.locher@dfinity.org}
\and
\IEEEauthorblockN{Yvonne-Anne Pignolet}
\IEEEauthorblockA{\textit{DFINITY} \\
yvonneanne@dfinity.org}
}

\begin{document}

\sloppy

\maketitle

\begin{abstract}
There is growing interest in providing programmatic access to the value locked in Bitcoin, which famously offers limited programmability itself. Various approaches have been put forth in recent years, with the vast majority of proposed mechanisms either building new functionality on top of Bitcoin or leveraging a bridging mechanism to enable  smart contracts that make use of ``wrapped'' bitcoins on entirely different platforms.

In this work, an architecture is presented that follows a different approach. The architecture enables the execution of Turing-complete Bitcoin smart contracts on the Internet Computer (IC), a blockchain platform for hosting and executing decentralized applications. Instead of using a bridge, IC and Bitcoin nodes interact directly, eliminating potential security risks that the use of a bridge entails. This integration requires novel concepts, in particular to reconcile the probabilistic nature of Bitcoin with the irreversibility of finalized state changes on the IC, which may be of independent interest.

In addition to the presentation of the architecture, we provide evaluation results based on measurements of the Bitcoin integration running on mainnet.
The evaluation results demonstrate that, with finalization in a few seconds and low execution costs, this integration enables complex Bitcoin-based decentralized applications that were not practically feasible or economically viable before.
\end{abstract}

\section{Introduction}
\label{sec:introduction}

Despite significant advances in blockchain technology over the past decade, 
Bitcoin is still the most well known decentralized platform, commanding a large share of the market across all projects in the blockchain space.
However, in stark contrast to other cryptocurrencies,
Bitcoin funds hardly appear in \emph{smart contracts}~\cite{szabo1996}, decentralized computer programs that automatically and deterministically execute their logic based on predefined conditions, due to the restricted expressiveness of Bitcoin's scripting language.
While the idea and possibility to create (simple) Bitcoin smart contracts has been around at least since 2011, this topic has garnered considerable attention in recent years and much effort has been expended with the goal to deliver powerful, general-purpose smart contracts for Bitcoin. Examples for such contracts include decentralized payroll or escrow systems, wallets, and applications for yield farming, lending and borrowing, as well as liquidity provision.

\begin{figure}[t]
\center
  \includegraphics[width=0.9\columnwidth]{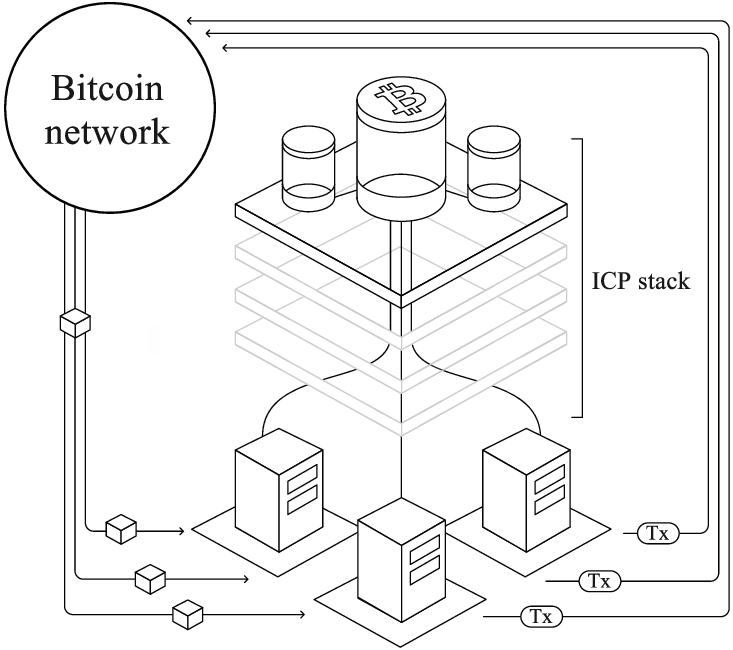}
    \caption{Architecture overview: IC node machines obtain blocks from the Bitcoin network and pass them through the ICP stack to the Bitcoin canister, which makes the Bitcoin blockchain state available to other canisters. Canisters can hold bitcoins natively and let node machines sign Bitcoin transactions on their behalf and forward them to the Bitcoin network.}
    \label{fig:intro}
   \vspace{-.4cm}
\end{figure}

Ideally, a smart contract platform for Bitcoin provides the following properties.
\begin{itemize}
\item \emph{Read and write}: Smart contract decisions can depend on the Bitcoin blockchain's state and write to it.
\item \emph{Security and trust}: Correct behavior as long as a supermajority of the entities involved adheres to the protocol.
\item \emph{Efficiency}: Typical user requests can be processed quickly (in the order of seconds, not minutes) at reasonable costs (fractions of cents, not several U.S. dollars).
\end{itemize}
Existing approaches to enable Bitcoin smart contracts, other than the presented architecture, fail to achieve all three of them.
Solutions extending Bitcoin directly are inherently slow and expensive.
Integrating Bitcoin functionality in another platform naturally requires reliable and secure access to the Bitcoin blockchain state. A customary approach to interconnect distinct blockchains is the use of so-called \emph{bridges}, which are (centralized or decentralized) third-party platforms that act as intermediaries. A typical use case of a bridge is to make an asset, such as a cryptocurrency, of one blockchain accessible on another. However, adding a bridge introduces a dependency, which needs to be trusted and increases the attack surface.
Since bridges have been marred by hacks that led to losses in the order of hundreds of millions~\cite{lee2023}, an integration without any reliance on bridges is preferable from a security point of view.
Proposed solutions without bridges either fail to provide a mechanism to read and write Bitcoin state directly, introduce security vulnerabilities, or are expensive and slow; see~\S\ref{sec:related_work} for a more in-depth discussion.

The architecture introduced in this paper provides support for Bitcoin smart contracts on the \emph{Internet Computer (IC)}\footnote{See \url{https://internetcomputer.org/}.}.
The motivation for this choice is manifold: First, the IC protocol suite (ICP stack) provides strong security guarantees as it can sustain a theoretically optimal upper bound on the number of malicious nodes and it maintains state integrity even under asynchrony (whereas \emph{liveness} requires partial synchrony). Moreover, it provides a smart contract execution environment that is efficient both in terms of speed and cost, and it allows smart contracts, called \emph{canisters} on the IC, to store large amounts of data. As we will see, our architecture leverages the latter capability to hold the sizable Bitcoin blockchain state.
Additionally, it offers a high degree of scalability, which is needed to support a large number of Bitcoin smart contracts.
Lastly, smart contracts interacting with the Bitcoin network must be able to create ECDSA or Schnorr signatures, which underpin the security of Bitcoin. The IC implements both threshold ECDSA~\cite{groth2022} and threshold Schnorr protocols that are secure under asynchrony, providing canisters with public keys for both schemes and the ability to sign arbitrary data under those keys.

In our architecture, IC nodes and Bitcoin nodes communicate directly, obviating the use of any additional infrastructure.
The \emph{Bitcoin adapter}, introduced in \S\ref{sec:bitcoin_adapter}, is the novel component that 
enables the direct interaction between the two networks.
This integration at the node level is used to send transactions---making use of the threshold-signing functionality exposed to smart contracts on the IC---to the Bitcoin network and to ingest Bitcoin blocks as shown in Figure~\ref{fig:intro}.
The Bitcoin blockchain state is maintained in a specific smart contract, called the \emph{Bitcoin canister}, presented in \S\ref{sec:bitcoin_canister}, which is an unprecedented approach given its sheer size.
Other canisters can learn about the current state of the Bitcoin blockchain and send transactions by interacting with the Bitcoin canister through a simple interface.

Although the IC has various features that facilitate cross-chain integrations, many challenges had to be tackled. In particular, Bitcoin and the IC exhibit many fundamental differences: Bitcoin only has probabilistic guarantees about the blockchain state and requires strong synchronicity assumptions regarding message transmissions. By contrast, the IC ensures that state change operations are never rolled back once \emph{finalized}. Moreover, the IC only requires sporadic periods of synchronicity for liveness as mentioned above. In short, the architecture must ensure that there is always deterministic agreement on the probabilistic state of the Bitcoin blockchain.
To this end, new notions of \emph{stability} of the Bitcoin blockchain state are introduced, which are used extensively in the architecture and are essential for its security.


The paper is structured as follows. Background information about both blockchains and definitions are provided in \S\ref{sec:definitions}. The main contribution of this work is the architecture itself, which is presented in detail in \S\ref{sec:architecture}.
The evaluation in \S\ref{sec:evaluation} consists of a discussion of security considerations (in \S\ref{sec:security}) to elucidate crucial aspects of the architecture design and a presentation of measurement results of the live system running on IC mainnet (in \S\ref{sec:measurements}).
Related work is summarized in \S\ref{sec:related_work} and \S\ref{sec:conclusion} concludes the paper.


\section{Background and Definitions}
\label{sec:definitions}


\subsection{Background on the Internet Computer (IC)}
\label{sec:background_ic}

We briefly present the IC Protocol (ICP) stack. The IC whitepaper~\cite{dfinity2022internet} provides a more comprehensive introduction.

\noindent\textbf{Objective.}
The IC aims to provide efficient multi-tenant, general-purpose, and secure computation in a decentralized and geo-replicated manner. 
In short, it is a blockchain-based platform for the execution of smart contracts called \emph{canisters}. Canisters are the smallest units bundling logic and state, allowing for parallel execution. In response to a message that a canister receives from a user or another canister, a canister executes its logic, which may trigger the transmission of messages, modification of its internal state, or the creation of other canisters. Moreover, canisters can schedule the execution of (parts of) their own code using timers, in contrast to most other smart contract platforms where execution can only be triggered by users.
All canisters are hosted on dedicated individually untrusted nodes running the ICP stack. 

\noindent\textbf{Subnets.} 
The nodes are partitioned into \emph{subnets}, each subnet providing blockchain-based state machine replication~\cite{schneider1990} for the set of canisters deployed on it and connecting to nodes in other subnets to enable communication between canisters hosted on different subnets.
Each node in a subnet runs all the canisters deployed on that subnet. There are subnets with 100,000+ canisters and subnets with just a handful of canisters. Subnets consist of 13-40 nodes spread across the Americas, Europe, and Asia.
Node providers and the location of their nodes are public, thus a high degree of fault tolerance can be achieved with relatively low replication factors.

\noindent\textbf{Adversarial model and fault tolerance.} 
The IC is designed for Byzantine fault-tolerance, i.e., faulty nodes may deviate in arbitrary ways from the IC protocol, 
accounting for bugs, outages, as well as outright malicious behavior by colluding nodes.
In any given subnet with $n = 3f+1$ nodes, at most $f$ nodes may be faulty. This is the highest number that
can be tolerated without additional assumptions on failures and message delivery~\cite{schneider1990,fischer1985}. 
The IC consensus protocol~\cite{camenisch2022} guarantees safety under asynchronous network conditions.

\noindent\textbf{Layers.} 
The ICP stack consists of four layers. 
The \emph{networking layer} disseminates protocol and user-generated messages. 
Message validation and ordering is handled by the \emph{consensus layer}.
The \emph{message routing layer} ensures that messages end up at the right canister and are scheduled deterministically.
The \emph{execution layer} triggers the deterministic execution of the canisters deployed on the IC and persists the state changes.


\begin{figure}[t]
\center
  \includegraphics[width=\columnwidth]{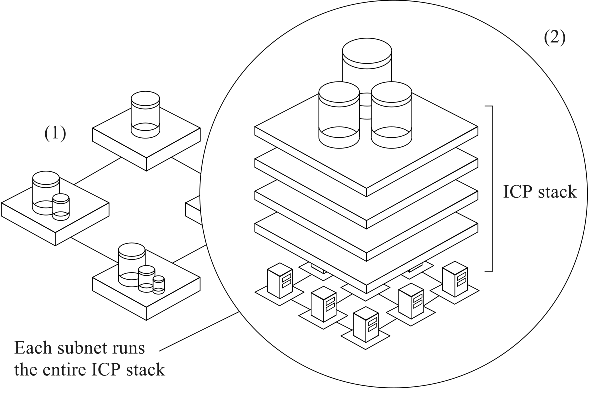}
    \caption{(1) The Internet Computer Protocol (ICP) partitions its nodes into mutually disjoint subnets. (2) The nodes of each subnet run a stack of protocols for state machine replication to execute canister smart contracts.
    }
    \label{fig:ic}
\end{figure}

\noindent\textbf{Deterministic finalization.} 
Once the IC consensus protocol~\cite{camenisch2022} reaches agreement on the next block (containing messages from users and canisters on other subnets) to be added to the subnet's blockchain, the block is considered \emph{finalized}. Since this decision is irreversible, the block content is guaranteed to be processed by the message routing layer and forks are impossible. Thus, finalization is \emph{deterministic}. 

Figure~\ref{fig:ic} depicts the structure of the Internet Computer, comprising multiple subnets with each subnet consisting of a replicated state machine run on multiple nodes.

\subsection{Background on Bitcoin}
\label{sec:background_bitcoin}

Bitcoin is a cryptocurrency operated by a decentralized peer-to-peer (P2P) network. Bitcoin transactions are used to transfer bitcoins from at least one party to one or more parties.
Each transaction produces a set of \emph{outputs}, which consist of a \emph{value}, i.e., the amount of bitcoin, and a \emph{locking script} defining the conditions to spend this output. The \emph{inputs} are simply previous outputs that are spent entirely in this transaction.

Since all inputs are fully consumed in a transaction, transferring the bitcoins to the outputs (minus a transaction fee), all spendable bitcoins are held in the \emph{unspent transaction output (UTXO) set}.
Thus, knowledge of the UTXO set suffices to determine the balance of any Bitcoin address.

Bitcoin transactions are processed in batches called \emph{blocks}, each block referencing a specific predecessor block. 
A block is only valid if its hash, interpreted as a large number, is at most a certain target value.
Let $\mathcal{H}$ denote the hash function used in Bitcoin, which is the SHA-256 algorithm applied to the block and then again to the resulting hash.
A numerically small $\mathcal{H}(b)$ implies a high level of \emph{difficulty} to find such a block by varying block metadata and shuffling transactions.\footnote{See \url{https://en.bitcoin.it/wiki/Difficulty}.}
The \emph{difficulty target} to create a block is adaptive and ensures that a large computational effort is required to create a block with 10 minutes between block creation on average, giving the network sufficient time to disseminate newly found blocks. The aforementioned transaction fee of every transaction in a block goes to the party that computed the block.

Multiple blocks may reference the same predecessor block, inducing a directed tree of blocks with the root being the \emph{genesis block}. Let $B$ denote the tree of all blocks available at some peer. Peers may have slightly different views and the following definitions always refer to a local view.
Given the set $B$ of blocks, the \emph{height} $h(b)$ of block $b \in B$ is the number of predecessor blocks, terminating at the genesis block $b_g$ with height $h(b_g) = 0$.
While there is only one path from any block to the genesis block, a block may have multiple successors.
Let $succ(b)$ denote the set of successors of block $b$. A depth function $d$ measures the maximum cumulative cost from a given block $b$ to any of the tips (i.e., leaves) that are connected to $b$ using a cost function $c: B \rightarrow \mathbb{R}$:\footnote{In graph theory, the definitions of \emph{height} and \emph{depth} are reversed. We use the terminology commonly found in Bitcoin literature.}

$$
d(b) \coloneqq \begin{cases}
  c(b) & succ(b) = \{\} \\
  c(b) + \max_{b' \in succ(b)} d(b') & \text{ otherwise}
\end{cases}
$$

Since a block $b$ is always associated with a block header $\beta$, the height and depth functions can be applied to block headers as well, i.e., $h(b) = h(\beta)$ and $d(b) = d(\beta)$ for a block $b$ and the associated block header $\beta$. For a given tree $B$, we further define that $d(B) \coloneqq d(b_g)$, which states that the depth of the blockchain corresponds to the depth of its genesis block.

We introduce two specific depth functions: The first function $d_c$ measures the maximum ``distance'' of a block $b$ to the tips, i.e., $c(b) \coloneqq 1$ for all $b \in B$.
The function $d_c$ is related to the concept of \emph{confirmations} in Bitcoin. Once a transaction appears in a block, the transaction is said to have one confirmation, and the confirmation count increases with each appended block. If a transaction in a block $b$ has $c$ confirmations, then $d_c(b) = c$. 
Let $w(b)$ denote the \emph{hash work} expended to create block $b$.
Technically, $w(b) > w(b')$ if $\mathcal{H}(b) < \mathcal{H}(b')$. The second function $d_w$ determines the maximum amount of hash work over all paths between $b$ and connected tips, i.e., $c(b) \coloneqq w(b)$ for all $b \in B$.
The function $d_w(b)$ determines the current blockchain, which is the chain of blocks from the genesis block to a tip that maximizes the sum of expended hash work. Formally, the \emph{current blockchain} is the chain of blocks $[b_0 = b_g, b_1, \ldots, b_k]$, $b_i \in succ(b_{i-1})$ for all $i \in \{1, \ldots, k\}$, such that $\sum_{i=0}^k d_w(b_i) = d_w(B)$.

More details on the Bitcoin protocol can be found online.\footnote{See \url{https://en.bitcoin.it/wiki/Protocol_documentation}.} 

\subsection{Novel Concepts}
\label{sec:concepts}

In addition to the definitions in \S\ref{sec:background_bitcoin}, we introduce our own concepts that address the lack of deterministic finality in Bitcoin. As mentioned before, the knowledge of the UTXO set suffices to determine the balance of any Bitcoin address; however, it is possible that a different chain (i.e., a \emph{fork}) becomes the current blockchain, invalidating all blocks on the former blockchain between its tip and the block where the two chains diverged. While such \emph{reorganizations} have become increasingly rare, it is nevertheless a risk, in particular for a smart contract that cannot be updated easily.

Since there is no deterministic finality, a weaker notion of ``stability'' is introduced, which can be motivated as follows. It is evident that a higher block depth (for either depth function) implies a higher chance that a block will persist. This is a necessary but not sufficient condition when considering that there can be competing forks growing at a similar rate: even if a block has a high block depth, another block might exist at the same height with a higher block depth. Thus, a block must have a sufficiently higher block depth than any other block at the same block height. Both conditions are captured in the following definition.

\begin{definition}[$\delta$-Stability] Given a depth function $d: B \rightarrow \mathbb{N}_0$, a block $b \in B$ is \emph{$\delta$-stable} if
\begin{enumerate}
 \item $d(b) \ge \delta$ and
 \item $\forall b' \in B \setminus \{b\}, h(b') = h(b): d(b) - d(b') \ge \delta$.
\end{enumerate}
\label{def:stability}
\end{definition}

For any $\delta>0$, the definition implies that there can be only one $\delta$-stable block at any height $h$, making it possible to extend the definition to block heights: height $h$ is $\delta$-stable if there is a block $b$ such that $h(b) = h$ and $b$ is $\delta$-stable.
It also follows that a $\delta$-stable block is $\delta'$-stable for any $\delta' \le \delta$. Given a certain depth function, the \emph{stability} of a block $b$ is the largest $\delta$ for which it is $\delta$-stable.

\begin{figure}[t]
\center
  \includegraphics[width=0.8\columnwidth]{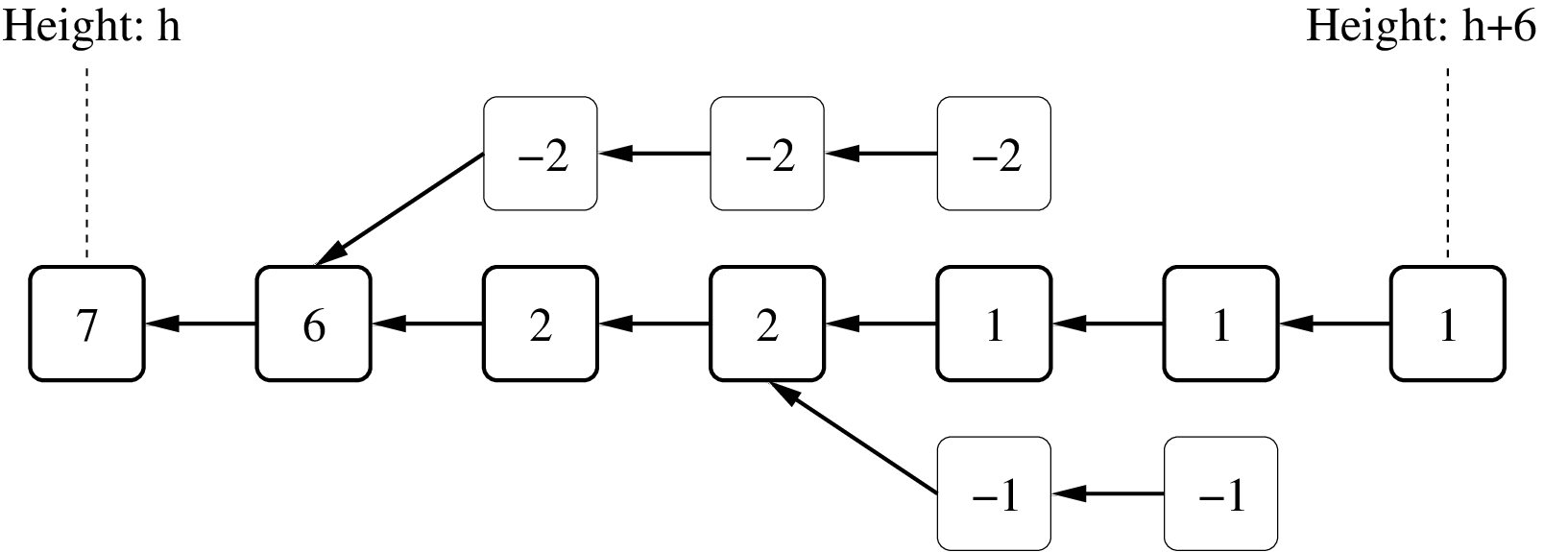}
    \caption{A blockchain with two forks is shown with the confirmation-based stability indicated inside each block.}
    \label{fig:stability}
\end{figure}

Definition~\ref{def:stability} can be instantiated with either depth function introduced in \S\ref{sec:background_bitcoin}, serving different purposes. When using the depth function $d_c$ to count confirmations, we call it \emph{confirmation-based ($\delta$-)stability}. More precisely, a transaction in a block $b$ is considered to have $c$ confirmations if the confirmation-based stability of $b$ is $c$.
By contrast, \emph{difficulty-based ($\delta$-)stability}, which applies the depth function $d_w$, is used to determine if a block is ``stable enough'' in the sense that it will persist with high probability. More precisely, the difficulty-based stability $d_w(b)$ is expressed relative to the difficulty of a particular block $b^*$, i.e., $d_w(b) / w(b^*)$, to be able to specify a threshold $\delta$ that is independent of current difficulties. We say that block $b$ is \emph{difficulty-based $\delta$-stable with respect to $b^*$ if $d_w(b) / w(b^*) \ge \delta$}.
Details about the usage of these concepts are presented in \S\ref{sec:bitcoin_canister}.

An example block tree is depicted in Figure~\ref{fig:stability} with blocks at the same height in the same horizontal position (and block heights increasing from left to right). The number in each block indicates its confirmation-based stability. The figure shows show that the stability of a block may stagnate even when its block depth increases. Another observation is that the stability of a block is negative when it is on a fork that is shorter than the longest chain.
While these concepts do not exist in Bitcoin, it is worth noting that they merely introduce conservative rules to deal with forks. In the absence of forks, the rules are equivalent to the standard definitions in the sense that confirmations are counted and the current blockchain is determined in the same way as customary in Bitcoin.

\section{Architecture}
\label{sec:architecture}



\subsection{Overview}
\label{sec:overview}

The functionality to implement Bitcoin smart contracts on the Internet Computer hinges upon two core building blocks. The first one enables IC nodes and nodes in the Bitcoin P2P network to interact directly: the \emph{Bitcoin adapter} is introduced at the networking level, running alongside the main IC process and connecting to nodes in the Bitcoin network. The Bitcoin adapter both ingests updates about the Bitcoin blockchain state (in the form of Bitcoin blocks) and transmits state changes to the Bitcoin network (in the form of Bitcoin transactions). As mentioned in \S\ref{sec:introduction}, this mechanism fundamentally differs from the commonly employed approach to overcome the siloed nature of blockchains using bridges.

The second building block provides an interface for canisters to read from and write to the Bitcoin blockchain.
Running on top of the IC stack, the \emph{Bitcoin canister} is responsible to keep track of the Bitcoin blockchain state.
Other canisters can interact with the Bitcoin canister through its API in order to get information about the Bitcoin blockchain as well as transmit Bitcoin transactions to update the Bitcoin blockchain state.

\begin{figure}[t]
\begin{center}
  \includegraphics[width=\columnwidth]{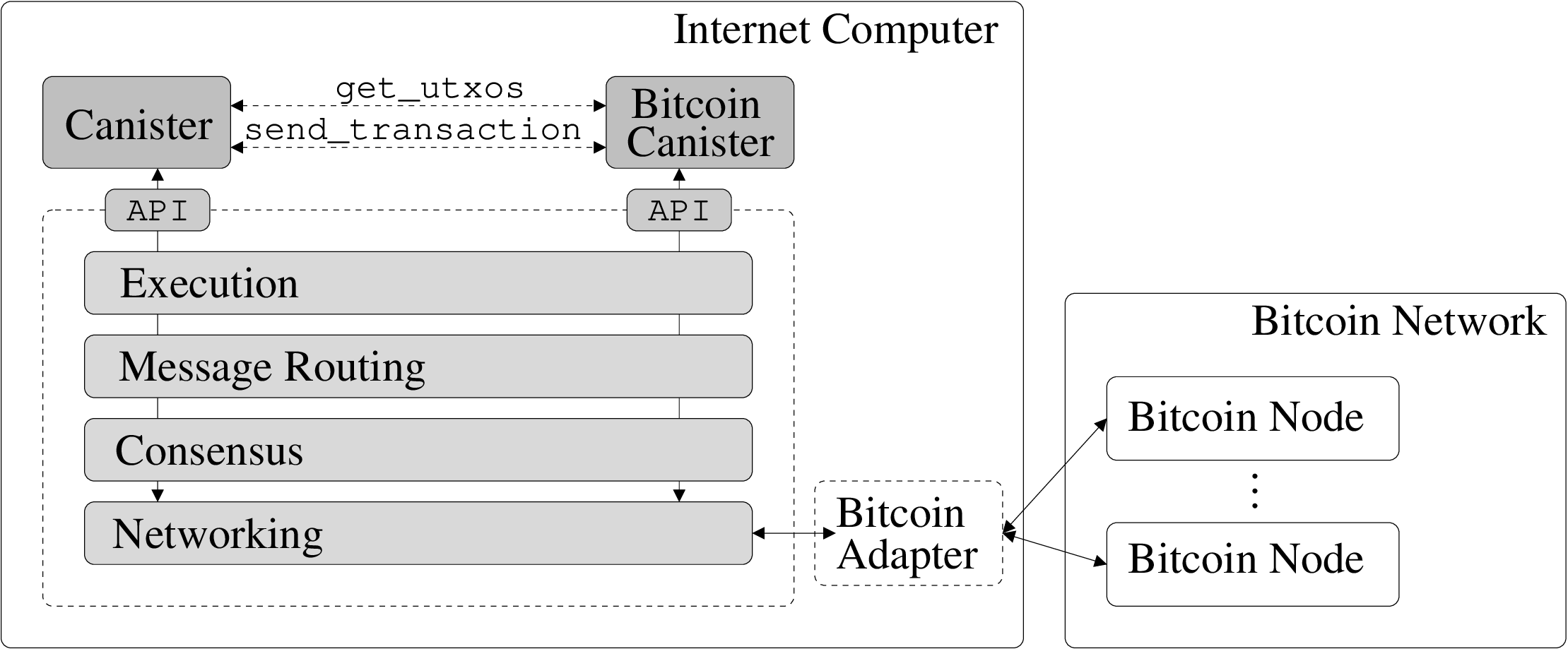}
    \caption{Overview of the Bitcoin integration on the IC.}
    \label{fig:overview}
\end{center}
\vspace{-.4cm}
\end{figure}

Figure~\ref{fig:overview} provides an overview of the architecture. The (sandboxed) Bitcoin adapter communicates with the Bitcoin network and the Bitcoin canister, interfacing at the ICP networking layer. The execution layer executes the Bitcoin canister, which processes requests from other canisters, e.g., requests to get the UTXOs of a Bitcoin address or to send a Bitcoin transaction to the Bitcoin network.

\subsection{Bitcoin Adapter}
\label{sec:bitcoin_adapter}

The objective of the Bitcoin adapter is to enable two-way communication between the Internet Computer and the Bitcoin network without any intermediaries.
The Bitcoin adapter is a sandboxed OS-level process that can be instantiated on every IC node. It communicates with the main IC process using standard inter-process communication. The source code, roughly 6000 lines of Rust code, resides in the IC repository.\footnote{See \url{https://github.com/dfinity/ic/tree/master/rs/bitcoin/adapter}.}

In some sense, the Bitcoin adapter is akin to a \emph{simplified payment verification (SPV)} Bitcoin client.\footnote{See \S8 of the Bitcoin white paper at \url{https://bitcoin.org/bitcoin.pdf}.} The main difference to a standard SPV client is that the Bitcoin adapter is geared towards providing information about the Bitcoin blockchain state to the Bitcoin canister, interacting with Bitcoin nodes using the Bitcoin P2P protocol and a bespoke protocol for the communication with the Bitcoin canister.

\noindent\textbf{Connectivity to the Bitcoin network.} 
Given a hard-coded list of DNS seed nodes---the same seed nodes as used by \texttt{bitcoind}---, the Bitcoin adapter attempts to connect to a configurable number $\ell$ ($=5$ on mainnet) of randomly chosen Bitcoin nodes using the following discovery process. At start-up, the Bitcoin adapter recursively requests IP addresses of Bitcoin nodes until the number of collected addresses reaches a certain upper threshold $t_u$. Once the threshold is reached, the Bitcoin adapter chooses Bitcoin nodes uniformly at random from the list of collected addresses and tries to establish $\ell$ connections. Whenever a connection is lost, a new random connection is established. If the number of available IP addresses drops below a lower threshold $t_l$, the Bitcoin adapter requests more addresses until its size reaches $t_u$ again.
The parameters $t_l$ and $t_u$ are set to $500$ and $2000$ for mainnet, $100$ and $1000$ for testnet, and $t_l = t_u = 1$ for regtest (for local testing). Experiments showed that these numbers (for mainnet and testnet) result in mostly disjoint sets of connected Bitcoin nodes at every Bitcoin adapter for common subnet sizes of 13 to 40 nodes.
The discovery process is skipped in regtest mode because the IP addresses are pre-configured.
If $t_u$ addresses cannot be collected, the Bitcoin adapter remains in the discovery state but provides its service to the Bitcoin canister as long as there is at least one active connection.

It is worth noting that IC nodes only have IPv6 addresses and therefore the Bitcoin adapter is restricted to interacting with Bitcoin nodes accessible over IPv6. Many Bitcoin nodes that have an IPv6 address are presumably \emph{dual-stacked}, i.e., they are themselves connected over IPv4 to other nodes.\footnote{Unfortunately, there is no reliable way to verify this assumption. It is reasonable to assume that the majority of node administrators do not configure connectivity exclusively over IPv6.}
We conjecture that the Bitcoin adapters connect to a fairly unbiased random sample of Bitcoin nodes.

\noindent\textbf{Interaction with the Bitcoin network.}
Once the Bitcoin adapter is connected to at least one Bitcoin node, it downloads the Bitcoin block headers, starting from the hard-coded genesis block header. For each obtained block header, the Bitcoin adapter verifies its validity, making sure that it is well-formed, the \texttt{hashPrevBlock} field points to a locally available block header, the \texttt{Bits} field contains the correct difficulty target, the block header hash satisfies this target, and the \texttt{Time} field contains a valid \emph{block timestamp}.\footnote{See \url{https://en.bitcoin.it/wiki/Block_timestamp}.}
Any block header violating some of these conditions is discarded.


The Bitcoin adapter does not perform any fork resolution, i.e., it accepts and stores any valid block header, possibly multiple block headers at the same height. Adding fork resolution logic would go against the goal of keeping the Bitcoin adapter as lightweight as possible, leaving the task of resolving forks and maintaining a correct state to the Bitcoin canister.

\noindent\textbf{Interaction with the Bitcoin canister.}
The Bitcoin adapter receives requests from the Bitcoin canister, which contain a specific block header $\bh^*$ and a set $\mathcal{A}$ of block headers at heights greater than $h(\bh^*)$ for which the Bitcoin canister has already obtained the blocks.
Additionally, each request contains a possibly empty list $T$ of Bitcoin transactions that are supposed to be transmitted to the Bitcoin network. Details about the request parameters are provided in \S\ref{sec:bitcoin_canister}.
Given this information, the Bitcoin adapter aims to return blocks that extend the Bitcoin canister's chain (or tree). Additionally, it sends a list of upcoming block headers, if any, to inform the Bitcoin canister that more blocks need to be synced.

	\begin{algorithm}[t]
\caption{Given $B_a$ and $\mathcal{B}_a$, process request $(\bh^*$, $\mathcal{A}$, $T)$.}
    \label{algo:bitcoin_adapter}
\begin{algorithmic}[1]
\FOR{$tx \in T$}
\STATE \texttt{transaction\_cache.add($tx$)}
\ENDFOR
\STATE $B \coloneqq \{\}$, $\mathcal{B} \coloneqq \{\}$, $\mathcal{N} \coloneqq \{\}$, $\bh_{cur} \coloneqq \bh^*$
\WHILE{ $\bh_{cur} \ne \bot$ \textbf{and} $|\mathcal{N}| <$ \texttt{MAX\_HEADERS} }
\IF{$\bh_{cur} \notin \mathcal{A}$ \textbf{and} $\bh_{cur}.prev \in \mathcal{A} \cup \mathcal{B}$}
\STATE $b_{cur} \coloneqq$ \texttt{get\_block($\bh_{cur}$, $B_a$, $h(\bh^*)$)}
\IF{$b_{cur} \ne \bot$ \textbf{and} \texttt{size($B$)} $<$ \texttt{MAX\_SIZE} \textbf{and} $|B| <$ \texttt{max\_blocks\_at\_height($h(\bh^*)$, $\mathcal{B}_a$)}}
\STATE $B \coloneqq B \cup \{(b_{cur}, \bh_{cur})\}$, $\mathcal{B} \coloneqq \mathcal{B} \cup \{\bh_{cur}\}$
\ENDIF
\ENDIF
\IF {$\bh_{cur} \notin \mathcal{A} \cup \mathcal{B}$}
\STATE $\mathcal{N} \coloneqq \mathcal{N} \cup \{\bh_{cur}\}$
\ENDIF
\STATE $\bh_{cur} \coloneqq$ \texttt{bfs\_next($\bh_{cur}$, $\mathcal{B}_a$)}
\ENDWHILE
\STATE \textbf{return} $[B, \mathcal{N}]$
\end{algorithmic}

	\end{algorithm}

The algorithm is defined more formally in Algorithm~\ref{algo:bitcoin_adapter}.
Let $B_a$ denote the set of locally available Bitcoin blocks and $\mathcal{B}_a$ be the local block header tree. While $\mathcal{B}_a$ is extended continuously, blocks are only added to or removed from $B_a$ when handling requests from the Bitcoin canister.
Every outbound transaction $tx \in T$ is put into a transaction cache. Transactions in this cache are advertised asynchronously to all connected Bitcoin nodes and transmitted upon request. A transaction is kept in the cache until transmitted to all connected peers or until it expires after a certain time, configured to 10 minutes. The expiry time ensures that memory is freed even if there are connected Bitcoin nodes that do not request advertised transactions.
Given that there is generally no guarantee that transactions are added to the Bitcoin mempool, this best-effort strategy is acceptable.

After handling the transactions in the request, the Bitcoin adapter
traverses its block header tree $\mathcal{B}_a$ in a breadth-first-search fashion (using the function \texttt{bfs\_next}), starting at the block header $\bh^*$.
The block and block header pairs that will be returned are collected in the set $B$, whereas set $\mathcal{B}$ is used to collect only the block headers in set $B$.
If the current block header $\bh_{cur}$ is not in the set $\mathcal{A}$ and its predecessor $\bh_{cur}.prev$ is in $\mathcal{A}$ or $\mathcal{B}$, i.e., $\bh_{cur}$ follows a block header for which the Bitcoin canister has the block or the corresponding block has been collected already, the block $b_{cur}$ is retrieved using the function \texttt{get\_block} if available. If $b_{cur}$ is not available, it is requested from the connected peers asynchronously so that the block may be served in the response to a future request. Otherwise, if the total size of blocks collected in $B$ does not exceed a certain maximum size (\texttt{MAX\_SIZE}=2MiB) and the total number of blocks in $B$ is at most a certain upper bound based on the height $h(\bh^*)$, then $(b_{cur},\bh_{cur})$ and $\bh_{cur}$ are added to $B$ and $\mathcal{B}$, respectively. Note that \texttt{MAX\_SIZE} is a soft limit in that a block that exceeds this size is still added to $B$. The maximum number of blocks is unbounded up to a certain hardcoded height and then it is reduced to $1$, i.e., only a single block may be returned. Returning multiple blocks speeds up the syncing process but returning only one block is preferable for security reasons as discussed in \S\ref{sec:security}.

If the block is not available, $\bh_{cur}$ is added to the set $\mathcal{N}$ of upcoming block headers before considering the next block header. Block headers are added to the set $\mathcal{N}$ as long as its size is below a configurable maximum size (\texttt{MAX\_HEADERS}$=100$ in production). If there are no more block headers to be processed ($\bh_{cur} = \bot$) or the set $\mathcal{N}$ reaches the maximum size, the Bitcoin adapter returns $[B, \mathcal{N}]$ to the Bitcoin canister.

\subsection{Bitcoin Canister}
\label{sec:bitcoin_canister}

The Bitcoin canister is a smart contract that provides the core Bitcoin functionality. It is a fairly complex smart contract at approximately 10,000 lines of Rust code.\footnote{See \url{https://github.com/dfinity/bitcoin-canister.}}
As shown in Figure~\ref{fig:overview}, the Bitcoin canister enables other canisters to read the Bitcoin blockchain state as well as writing to it.

Instead of storing the entire Bitcoin blockchain, which would require several hundreds of gigabytes of storage space, the Bitcoin canister merely stores the UTXO set. As discussed in \S\ref{sec:background_bitcoin}, the UTXO set is sufficient to derive the current balance of any Bitcoin address while reducing the space requirement. The downside of this approach is that fork resolution becomes challenging as it is hard to undo updates to the UTXO set.

\noindent\textbf{State management.} 
The Bitcoin canister uses the concepts introduced in  \S\ref{sec:concepts} to deal with the probabilistic finality in Bitcoin.
Concretely, a block is considered \emph{stable} if it is difficulty-based $\delta$-stable for $\delta = 144$ on mainnet. As there are $144$ Bitcoin blocks on average per day, unless the difficulty target changes significantly, a block becomes stable after approximately one day. Such blocks are expected to persist.

The block header $\bh^*$ at the greatest height that is considered stable is called the \emph{anchor}. 
The Bitcoin canister stores the UTXO set $U$ of the Bitcoin blockchain up to and including the anchor height, whereas \emph{all} block headers
are stored in a directed tree $\mathcal{T}$. 
While conceptually the UTXO set $U$ is simply the set of all unspent transaction outputs, the implementation uses a data structure with Bitcoin addresses as the index for an efficient retrieval of all UTXOs associated with an address.

For heights greater than $h(\bh^*)$, the Bitcoin canister additionally stores the corresponding blocks.
The block for block header $\bh$, $h(\bh) > h(\bh^*)$, is denoted by $b(\bh)$. If the block is not available, then $b(\bh) = \bot$.
Since unstable blocks are stored separately, both the UTXO set $U$ and the unstable blocks must be considered to determine the current UTXOs of an address.
Consequently, the computational complexity to retrieve all UTXOs or compute the balance of an address grows linearly with the parameter $\delta$. Hence, there is a trade-off between the computational complexity and security as a larger $\delta$ makes it less likely that blocks at heights lower than $h(\bh^*)$ are affected by a block reorganization. Note that the Bitcoin canister can cope with any block reorganization at heights greater than $h(\bh^*)$ automatically. Conversely, a reorganization at a lower height would require a manual canister upgrade as the UTXO set would need to be updated.
Setting $\delta = 144$ is a conservative choice, aiming for high security, i.e., a low risk of being affected by a block reorganization, while still guaranteeing a fast processing of requests.


\noindent\textbf{Interaction with the Bitcoin adapter.} 
The Bitcoin canister periodically requests updates from the Bitcoin adapter by sending a message containing the anchor $\bh^*$ and a list $\mathcal{A}$ of the block headers in $\mathcal{T}$ for which it has already received the corresponding blocks, i.e., $\mathcal{A} \coloneqq \{\bh \in \mathcal{T} \;|\; b(\bh) \ne \bot\}$.
The request also contains the set $T$ of Bitcoin transactions that are to be advertised to the Bitcoin network.

\begin{algorithm}[t]
  \small
\caption{Given $U$, $\mathcal{T}$ and $\bh^*$, process response ($B$, $\mathcal{N}$).}
    \label{algo:bitcoin_canister}
\begin{algorithmic}[1]
\FOR {$(b,\bh) \in B$}
\IF {\texttt{is\_valid($b$, $\bh$)} \AND \texttt{is\_valid($\bh$, $\mathcal{T}$)}}
\STATE \texttt{append($\bh$, $\mathcal{T}$)}
\STATE $b(\bh) \coloneqq b$
\STATE $B_{next} \coloneqq \{ b(\bh) \ne \bot \;|\; \bh \in \mathcal{T}, h(\bh) = h(\bh^*)+1\}$
\STATE $b_{next} \coloneqq \arg\max_{b \in B_{next}} d_w(b)$
\WHILE {$d_w(b_{next}) / w(\bh^*) \ge \delta$}
\STATE $\bh^* \coloneqq$ \texttt{get\_header($b_{next}$)}
\STATE \texttt{process\_block($U$, $b_{next}$)}
\STATE \texttt{remove\_blocks($\mathcal{T}$, $B_{next}$)}
\STATE $B_{next} \coloneqq \{ b(\bh) \ne \bot \;|\; \bh \in \mathcal{T}, h(\bh) = h(\bh^*)+1\}$
\STATE $b_{next} \coloneqq \arg\max_{b \in B_{next}} d_w(b)$
\ENDWHILE
\ENDIF
\ENDFOR
\FOR {$\bh \in \mathcal{N}$}
\IF {\texttt{is\_valid($\bh$, $v$)}}
\STATE \texttt{append($\bh$, $\mathcal{T}$)}
\ENDIF
\ENDFOR
\STATE $\mathcal{A} \coloneqq \{\bh \in \mathcal{T} \;|\; b(\bh) \ne \bot\}$
\STATE $synced \coloneqq \max_\mathcal{\bh \in \mathcal{T}} h(\bh)- \max_{\bh \in \mathcal{A}} h(\bh) \le \tau$
\end{algorithmic}

\end{algorithm}

The response of the Bitcoin adapter is a set $B$, containing pairs consisting of blocks and their headers, and a set $\mathcal{N}$ of block headers.
Algorithm~\ref{algo:bitcoin_canister} shows how the Bitcoin canister handles a response $(B, \mathcal{N})$ received from the Bitcoin adapter.

For each pair $(b,\bh)\in B$, it is verified that both $b$ and $\bh$ are valid. To this end, the Bitcoin canister performs the same checks on the block headers as the Bitcoin adapter (see \S\ref{sec:bitcoin_adapter}). A block $b$ is valid if the corresponding block header $\bh$ is valid and $b$ is well-formed, $\bh$ points to a predecessor for which the block is available, and the Merkle tree root hash of $b$ is the hash in $\bh$.
Furthermore, it is verified that $\bh$ is a valid extension of the block headers in $\mathcal{T}$.
Note that the validity of the transactions is not verified. The Bitcoin canister relies on the proof of work that goes into the blocks and the verification of the blocks in the Bitcoin network.
Transaction validation is omitted because a bug in the transaction verification logic is deemed a bigger security risk than relying on the vetting of blocks in the Bitcoin network. Moreover, the notion of $\delta$-stability adds a level of protection against forks.

If both $b$ and $\bh$ are valid, $\bh$ is added to $\mathcal{T}$ and $b$ is stored. Next, it is verified whether the addition of $b$ renders any block $b_{next}$ at height $h(\bh^*)+1$ difficulty-based $\delta$-stable with respect to $b(\bh^*)$ in which case the block header of block $b_{next}$ becomes the new anchor. Whenever the anchor changes, the UTXO set $U$ and the tree $\mathcal{T}$ must be updated as well: The UTXO set $U$ is updated by processing the transactions in $b_{next}$ and then block $b_{next}$ is discarded (i.e., $\beta(b_{next}) \coloneqq \bot$).
These steps are performed in a loop as more than one block may become stable with the addition of a single block.
Unlike blocks, block headers are kept forever; however, if there are multiple block headers at a stable height, all but the single stable block header are removed from $\mathcal{T}$.
Finally, all validated block headers in the received set $\mathcal{N}$ are appended to $\mathcal{T}$.

Since it is risky to provide outdated information about the blockchain state, the Bitcoin canister only responds to requests if the maximum height in $\mathcal{T}$ does not exceed the maximum height of available blocks by more than a parameter $\tau$ ($=2$ in production). If this condition is not satisfied, the Bitcoin canister replies with an error to each request. This mechanism explains the addition of the block header set $\mathcal{N}$ to the response of the Bitcoin adapter: it informs the Bitcoin canister about missing blocks in a tamper-proof manner.

\noindent\textbf{Application programming interface.} 
The goal is to provide a simple, low-level interface that is powerful enough to enable complex Bitcoin smart contracts.
The two core functions to read from and write to the Bitcoin blockchain state are called \texttt{get\_utxos} and \texttt{send\_transaction}, respectively. The API contains several additional functions, such as a convenience function to get the balance of an address (\texttt{get\_balance}), which are omitted for the sake of brevity.

When calling the \texttt{get\_utxos} endpoint, a Bitcoin address and network (mainnet, testnet, or regtest) must be specified. Optionally, a \emph{filter} can be provided as well, which is either a certain \emph{number of confirmations} or a request for a specific \emph{page}. The response from the Bitcoin canister consists of a set of UTXOs of the requested Bitcoin address for the given network, the hash and height of the block header at the tip of the considered chain, plus a \emph{next page reference}, which is non-empty if the response does not contain all UTXOs. This pagination mechanism is required for addresses that hold a large number of UTXOs. The UTXOs are returned sorted by block height in descending order, ensuring the correctness of the pagination mechanism.

As described in \S\ref{sec:concepts}, the blockchain is defined as the chain that maximizes $d_w(b_g)$, where $b_g$ denotes the genesis block. If the request specifies a minimum of $c$ confirmations, only confirmation-based $c$-stable blocks are considered.
It is important to note that requests for $c > \delta = 144$ are rejected as the returned set of UTXOs may not be correct in that case: It is possible that the response is missing some outputs because transactions that spend these outputs should not be considered for the given choice of $c$; however, this information is not contained in the UTXO set.


The \texttt{send\_transaction} endpoint takes two parameters: the serialized Bitcoin transaction and the target network.
After performing basic checks to ensure that the received bytes encode a syntactically correct transaction
, the transaction is included in the set $T$ of transactions that is forwarded to the Bitcoin adapter as part of its regular update requests.


\section{Evaluation}
\label{sec:evaluation}


\subsection{Security Considerations}
\label{sec:security}

The availability and integrity of the Bitcoin integration functionality are the two crucial security considerations that have guided the design of the architecture. The purpose of this section is to illustrate which attack scenarios have been considered and how they are mitigated.

We start by formally defining the assumptions with respect to the attacker's control over IC and Bitcoin nodes. 
Specifically, we consider a powerful attacker that can simultaneously control a large fraction of all Bitcoin nodes and the Internet Computer nodes. Moreover, the attacker has a large hash power available so that it can create arbitrary Bitcoin blocks for the same difficulty target as in Bitcoin mainnet, albeit at a lower rate than the Bitcoin network itself. Rather than specifying concrete parameters, we define the extent of the attacker's power in the course of this section, together with explanations where necessary.

The Bitcoin canister runs in a subnet comprising $n$ nodes, where the Bitcoin adapter on each node connects to $\ell$ randomly selected Bitcoin nodes.
We make the standard assumption that the attacker controls less than $n/3$ of these nodes.
Let $\varphi$ be the fraction of Bitcoin nodes under the attacker's control. In order to cut off the Bitcoin canister from the Bitcoin network, thereby preventing updates to the state in the Bitcoin canister, the attacker must control a large share of all Bitcoin nodes, which is deemed infeasible for practical values of $n$ and $\ell$. 

\begin{definition}\label{def:connectivity} For any subnet size $n$ and number $\ell$ of links from Bitcoin adapters to Bitcoin nodes, the fraction $\varphi$ of corrupted Bitcoin nodes is upper bounded by 
$\varphi \ll n^{-1/\ell}$
(1).

\end{definition}

Given the large number of Bitcoin nodes, this assumption easily holds for parameters used in practice, i.e., for $n=13$ and $\ell=5$, the requirement is that $\varphi \ll 0.6$. If a lower constant bound on $\varphi$ is required, it is always possible to set $\ell \in \Theta(\log(n))$, undershooting $\varphi$ by any constant factor.

The Bitcoin canister makes progress as long as \emph{at least one} Bitcoin adapter is connected to at least one correct Bitcoin node. 
Definition~\ref{def:connectivity} 
implies that \emph{every} Bitcoin adapter connects to a correct Bitcoin node with high probability.

\begin{lemma}\label{lemma:connectivity}
If each Bitcoin adapter connects to $\ell$ Bitcoin nodes uniformly at random, then every Bitcoin adapter connects to a correct node with overwhelming probability.
\end{lemma}
\begin{proof}
The probability that a Bitcoin adapter connects only to corrupted Bitcoin nodes is $\varphi^\ell$. Thus, the probability
that it is eclipsed is
$
 1- (1 - \varphi^\ell)^n \approx 1 - e^{-n\varphi^\ell} \stackrel{(1)}{\approx}1 - 1 =0.
$
 \end{proof}
 
If the Bitcoin canister is almost certain to remain connected to correct Bitcoin nodes, an attacker can only try to corrupt the state of the Bitcoin canister. As discussed in \S\ref{sec:bitcoin_adapter}, the Bitcoin adapter only accepts \emph{valid} block headers and blocks, which makes it impossible for an attacker to flood the Bitcoin canister with invalid data.

Since the Bitcoin canister does not verify that the spending conditions of transactions are satisfied, an attacker can attempt to feed the Bitcoin canister valid blocks with manipulated transactions. However, this is a costly attack as valid blocks require a certain proof of work.
Quite generally, there is always a chance that the attacker mines a block before other miners, even if the attacker's hash rate is substantially smaller than the total hash rate. As a result, it is necessary to wait for a certain number of confirmations, reducing the risk of a reorganization that removes the block containing the transaction in question.

By the same principle we define that any critical actions by smart contracts that depend on the Bitcoin canister require $c^*$ confirmations, where $c^*$ is large enough so that the attacker is not able to create a fork with a height that exceeds the ``real'' blockchain's height by $c^*$ at the same difficulty level, i.e., the attacker may create a longer chain only at a reduced difficulty.

\begin{definition}\label{def:hash_rate}
Given blockchain $B$ of height $h$ and a constant $c^*$, the attacker's hash rate is bounded so that the height $h'$ of the attacker's blockchain $B'$ is less than $h+c^*$ or $d_w(B') < d_w(B)$ at all times with overwhelming probability.
\end{definition}

This is a reasonable assumption as otherwise the attacker can launch (double-spend) attacks against any service that works with Bitcoin, including centralized exchanges.

An attacker may attempt to corrupt the state of a smart contract by manipulating the state of the Bitcoin canister such that the attacker's fork is considered to have $c^*$ confirmations.
This is infeasible under the assumption of Definition~\ref{def:hash_rate}.

\begin{lemma} The probability of a state corruption for Bitcoin services on the IC requiring $c^*$ confirmations is negligible.\label{lemma:state_corruption}
\end{lemma}
\begin{proof}
We assume that the attacker has the means to send any (valid) block to the Bitcoin canister. Furthermore, we assume that there is a corrupting transaction in a block $b'$ at a height $h'$ on a forked chain $B'$ created by the attacker.

If this chain has a maximum height of $h_{max}+c^*$ instead of the maximum height $h_{max}$ of the real blockchain $B$, then $d_w(B') < d_w(B)$ due to Definition~\ref{def:hash_rate}. Since difficulty-based stability is used to identify the current blockchain, the Bitcoin canister ignores the attacker's fork and consequently does not consider the corrupting transaction in any response.

If $d_w(B') > d_w(B)$, the attacker's chain has a maximum height of less than $h_{max} + c^*$. If $b'$ has at least $c^*$ confirmations on $B'$, it follows that $h' < h_{max}$, i.e., there is a block $b$ at height $h'$ on the real blockchain $B$.
Moreover, we have that $d_c(b') - d_c(b) < c^*$, implying that $b'$ is not confirmation-based $c^*$-stable. As confirmation-based stability determines the number of confirmations, the Bitcoin canister never reports $c^*$ or more confirmations for the corrupting transaction.
\end{proof}

Given a conservative upper bound on the hash power of the attacker as specified in Definition~\ref{def:hash_rate}, Lemma~\ref{lemma:state_corruption} illustrates how the notion of stability helps to overcome the probabilistic nature of Bitcoin. If $\delta$ is chosen large enough, the attacker effectively requires a commanding share of all hash power to launch an attack that causes a state corruption. In this scenario, the attacker has the power to undermine the integrity of most Bitcoin services. As mentioned before, $\delta = 144$ is a conservative choice, which means that the attacker must create $144$ blocks \emph{more} than the Bitcoin network over any period of time to corrupt the Bitcoin canister state.

While the state of the Bitcoin canister is considered safe when it is running and fully synced, there is added risk when syncing the Bitcoin blockchain, either initially or after an extended downtime of the Bitcoin canister, causing the Bitcoin canister to be out of sync. It is important to note that the latter situation has never occurred but security measures for this possibility are in place nonetheless.
If the attacker knows that the Bitcoin canister will not sync beyond a specific block height $h^*$ until time $t$ and $t$ lies sufficiently far in the future, the attacker can use the time to build a fork of significant length starting at height $h^*+1$ even if the attacker's hashing power is a small fraction of the total hashing power of the Bitcoin network.\footnote{For example, at 1\% of the total hashing power, the attacker can mine 10 blocks in expectation in a week at the difficulty level of the Bitcoin network.}
The attacker would then try to get $c^*$ blocks accepted by the Bitcoin canister \emph{before} it learns about the blocks on the real blockchain.

This risk is mitigated by the correct Bitcoin adapters, which send the set $\mathcal{N}$ of block headers at greater heights in their responses, ensuring that the Bitcoin canister does not enter the synced state prematurely.
Thus, even if the attacker manages to serve the blocks on the fork, the Bitcoin canister does not act upon them. Once the Bitcoin canister is synced, the attacker's fork will not be longer by $c^*$ blocks by assumption.

The risk is greater after a (hypothetical) downtime of the Bitcoin canister.
Since less than a third of the nodes in the subnet might be malicious, the attacker may attempt to use these nodes to ingest a fork of length at least $c^*$ as soon as the Bitcoin canister is operational again. The consensus algorithm of the IC ensures that the next block maker cannot be predicted, and it is the block maker that proposes the IC block containing the Bitcoin block. Thus, the attack succeeds if malicious IC nodes are chosen as the block makers, forwarding the blocks on the attacker's fork to the Bitcoin canister while claiming that there are no further block headers (i.e., $\mathcal{N} = \{\}$).

\begin{lemma} If the attacker controls $f < n/3$ nodes of the subnet, the probability of a state corruption for the IC Bitcoin services requiring $c^*$ confirmations after downtime is $<3^{-c^*}$.
\end{lemma}
\begin{proof}
Since the Bitcoin canister only accepts one block at a time, each malicious block maker can only deliver one block on the fork to the Bitcoin canister.
If there is any round where a correct IC node is chosen as the block maker, it will provide the list $\mathcal{N}$ of correct block headers as each Bitcoin adapter is connected to correct Bitcoin nodes due to Lemma~\ref{lemma:connectivity}. Since the maximum height of the attacker's fork does not exceed the maximum height of the real blockchain by more than $c^*-1$ by assumption,
the attack only succeeds if malicious block makers are chosen $c^*$ times in a row. As the attacker controls less than $n/3$ IC nodes, the claim follows.
\end{proof}

The probability of $3^{-c^*}$ may appear high for a customary choice of $c^*$ but the attack also requires the corruption of many IC nodes in addition to a predictable downtime of the Bitcoin canister. As a result, such an attack is deemed highly unlikely.

\subsection{Measurements}
\label{sec:measurements}

This section studies the resource consumption of the Bitcoin canister deployed on IC mainnet and the interaction with the real Bitcoin network, in terms of the state size and the number of \emph{WebAssembly instructions} required to maintain the state and handle requests, as well as the latency experienced by users when interacting with the Bitcoin canister. We omit a discussion of throughput capacity due to space constraints and the fact that capacity can be increased linearly on demand by hosting Bitcoin canisters on more subnets.

\ifdefined\SRDS
\else
	\begin{figure}[t]
	\includegraphics[width=\columnwidth]{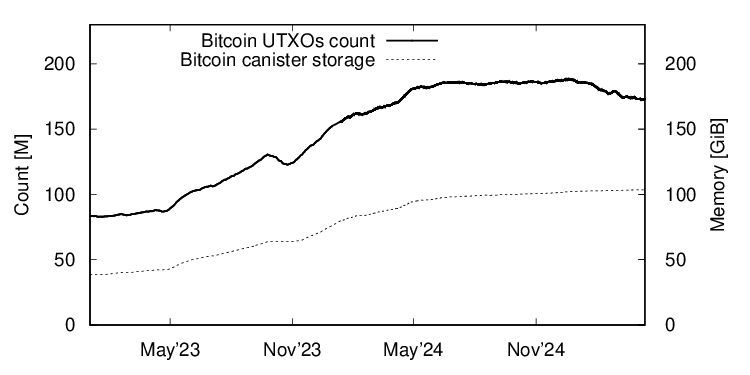}
	 \caption{The growth of the UTXO set and the Bitcoin canister space consumption is shown over the span of two years.}
	 \label{fig:memory_utxos}
\end{figure}
\fi

\begin{figure*}[t]
	\includegraphics[trim=14 8 8 8,clip=true,width=0.5\textwidth]{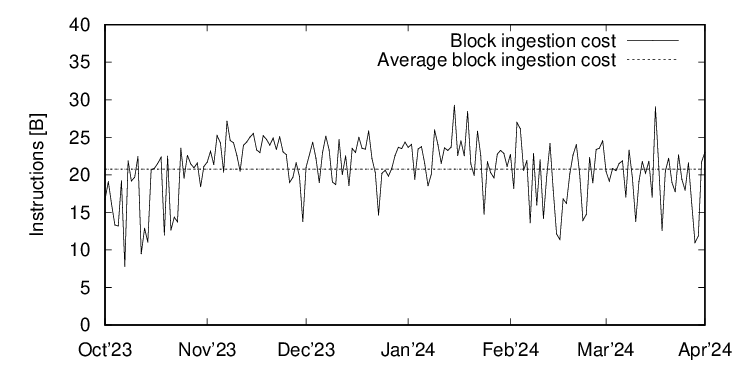}
	\includegraphics[trim=14 8 8 8,clip=true,width=0.5\textwidth]{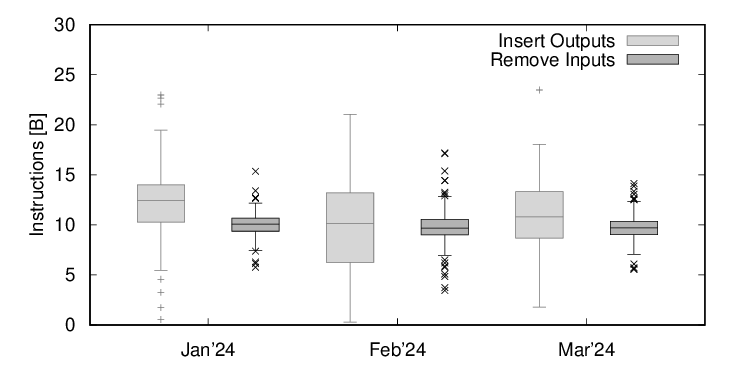}
	\vspace{-.2cm}	
			\caption{ 
				\textit{Left:} Number of instructions and the average of 21.6b instructions per block when ingesting blocks over a six-month period (Oct 2023--Apr 2024).
				\textit{Right:} Number of instructions executed for output insertions and input removals in Q1 2024.}
\label{fig:blocks}
\end{figure*}



\noindent\textbf{Storage and block ingestion.} 
The Bitcoin canister stores the whole UTXO set. Thus, its storage requirement grows linearly with the size of the UTXO set. By the end of March 2025, the Bitcoin canister reached a size of more than 103 GiB, storing more than 170 million UTXOs as shown in
\ifdefined\SRDS
	Figure~\ref{fig:blocks} (left).
\else
	Figure~\ref{fig:memory_utxos}.
\fi
The number of instructions executed for block ingestion varies with the size of the block as evident in 
\ifdefined\SRDS
	Figure~\ref{fig:blocks} (right).
\else
	Figure~\ref{fig:blocks}. 
\fi
It is typically in the order of 20 billion instructions, with roughly half of them used for output insertions and input removals, respectively.

\noindent\textbf{Latency and cost of handling requests.} 
In order to evaluate the time required to handle user requests and the resource consumption measured in instructions, we conducted the following experiments on IC mainnet. 
We selected 1000 random bitcoin addresses that appeared in blocks in Q1 2024 with a positive balance. 
The distribution of their UTXO set sizes is skewed with 517 having fewer than 50 UTXOs, 159 addresses returning sets of 50-199 UTXOs, 113 addresses returning 200-999 UTXOs, and 211 address having 1000 or more UTXOs.
For each of these addresses, we sent \emph{replicated} balance and UTXO requests to the Bitcoin canister and measured the time to receive a response and the resource consumption. In addition, we sent balance and UTXO \emph{query} requests. Every response to a query request comes from a single, randomly selected node on the subnet and therefore cannot be fully trusted. By contrast, the responses for the replicated requests are threshold-signed by more than two thirds of the nodes of the subnet.
We repeated the experiments with different address sets for the same time period and obtained similar results.

On average, replicated requests take below 10s to be answered, with the minimum around 7s and a $90^\mathit{th}$ percentile of 18s. 
For queries, which neither go through consensus nor require communication across subnets, the median time to get a balance or UTXOs is about 220ms and 310ms, respectively, and 90\% of the response times are below 0.5s and 2.5s, with considerably higher variance for UTXO requests.


Figure~\ref{fig:utxo_size} (left) illustrates that the response time for UTXO and balance requests depends on the size of the UTXO set. This observation is more visible for queries than replicated requests because in the latter the response time is dominated by the time necessary for consensus and other protocol overhead.
We derived the number of instructions executed from the cost of replicated UTXO requests.
Figure~\ref{fig:utxo_size}  (right) shows that the number of executed instructions varies between $5.84\cdot10^6$ and $4.76\cdot10^8$ with a clear correlation to the size of the response. The bifurcation in the figure is due to the distinction between stable and unstable blocks: UTXOs in unstable blocks can be fetched more quickly compared to UTXOs that have been migrated to the large UTXO set.


\begin{figure*}[t]
	\centering
	\begin{subfigure}{.32\textwidth}
	\hspace{-.2cm}
	  \includegraphics[trim=14 14 12 12,clip=true,width=\textwidth]{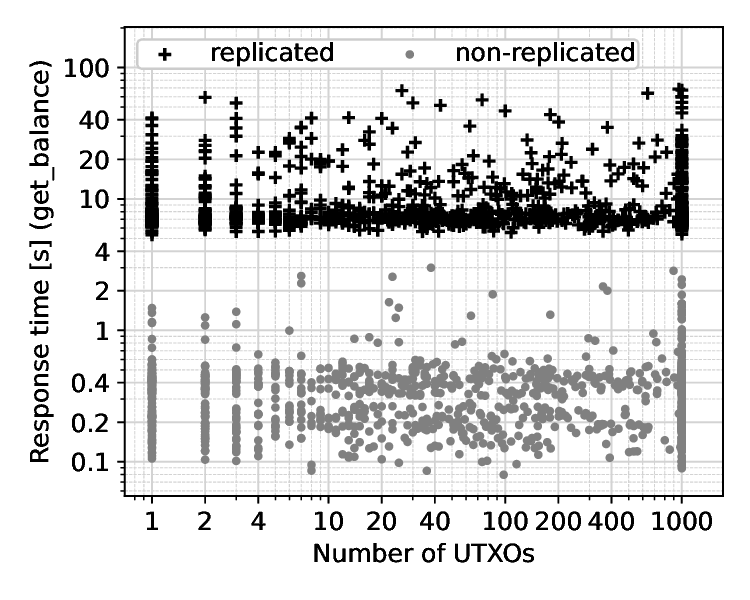}
	  \end{subfigure}
	  \hspace{.1cm}
	\begin{subfigure}{.32\textwidth}
	  \centering
	  \includegraphics[trim=14 14 12 12,clip=true,width=\textwidth]{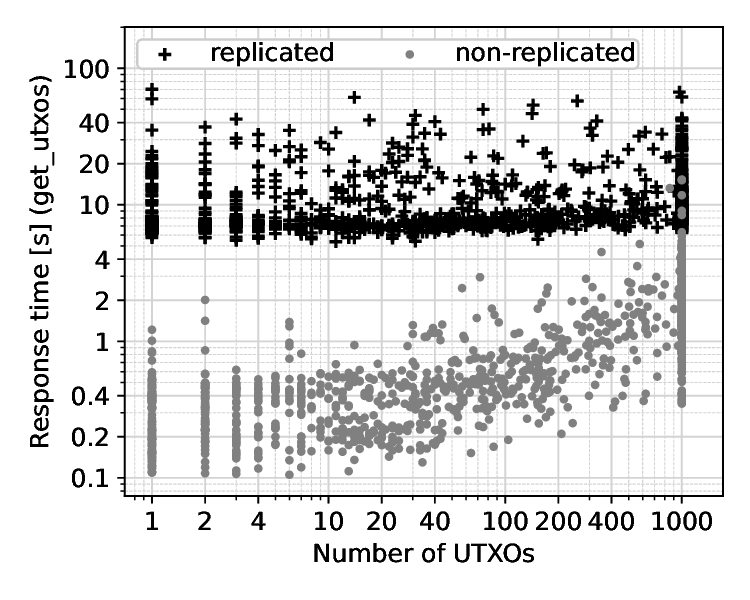}
	  \end{subfigure}
	  \hspace{.1cm}
	\begin{subfigure}{.32\textwidth}
	  \centering
	  \includegraphics[trim=14 14 12 12,clip=true,width=\textwidth]{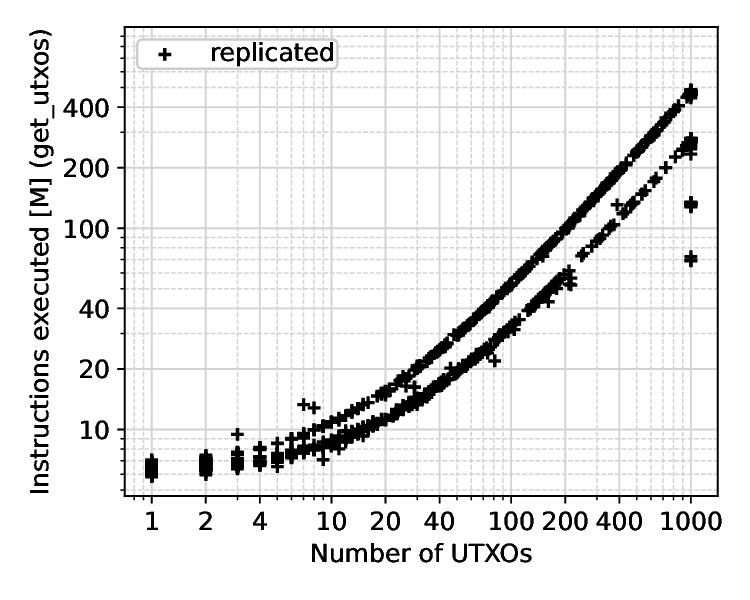}
	  \end{subfigure}
	\caption{\textit{Left/center:} Response time for replicated and non-replicated requests for get\_balance and get\_utxos respectively. \textit{Right:} Instructions executed for replicated UTXO requests.}
	\label{fig:utxo_size}
	\end{figure*}

At the current exchange rate, approximately 35,000 (1,500) requests for balances (UTXOs) can be made for 1 U.S.\ dollar. By comparison, the average fee for a single Bitcoin transaction was about 1-2 U.S.\ dollars at the end of 2024.

%

\section{Related Work}
\label{sec:related_work}

The idea of Bitcoin smart contracts can be traced back at least to 2011, where, e.g., escrow and assurance contracts using Bitcoin were proposed.\footnote{See the history of \url{https://en.bitcoin.it/wiki/Contract}.}
In general, work on this topic can be partitioned into three different approaches. The first approach is to build smart contracts strictly using the power of the Bitcoin scripting language. This line of work includes custom languages that facilitate the specification of smart contracts and their compilation into valid Bitcoin transactions. A related approach is to extend Bitcoin's scripting language to boost its expressiveness.
The last approach is to run some form of Bitcoin smart contracts on a different platform. We will discuss work for each of these approaches separately.

One of the first papers about Bitcoin smart contracts introduced \emph{timed commitments}, which can be utilized, e.g., to implement 
lotteries~\cite{andrychowicz2014}. Subsequently, more general multi-party computation approaches~\cite{andrychowicz2014b,bentov2014,kumaresan2014} were explored. 
TypeCoin~\cite{crary2015} models updates of a state machine as affine logic propositions, with liveness only guaranteed when all parties cooperate. Other languages were proposed that compile transactions of a protocol to Bitcoin scripts~\cite{atzei2018,oconner2017b}. BitML~\cite{atzei2019,bartoletti2018} can be used to write entire protocols, which are then directly compiled to a set of Bitcoin transactions. 

The second approach concerns extensions to the scripting language, typically in the form of new opcodes. A \emph{covenant} is a primitive that allows transactions to restrict how the value they transfer is used in the future~\cite{moser2016, oconnor2017}. Covenants can be used to implement, e.g., vaults and \emph{poison transactions} to penalize double-spending attacks.  Recursive covenants make it possible to implement a state machine that stores a certain state through a series of transactions. 
There are several other proposals introducing opcodes to enable advanced smart contracts~\cite{delgado2020,kumaresan2015,miller2017}. An alternative is to 
introduce malleability of transaction inputs~\cite{bartoletti2017}.
Both the first and second approach suffer from Bitcoin's high costs and latency.

This work belongs to the third category of running Bitcoin smart contracts on a different platform. While there is little literature on this approach, numerous other blockchain-based platforms have been built in recent years that aim to offer Bitcoin smart contracts, such as Stacks\footnote{See \url{https://stx.is/nakamoto}.}, Rootstock~\cite{lerner2022},  THORChain\footnote{See \url{https://thorchain.org/}.}, and WBTC\footnote{See \url{https://wbtc.network/}.} among others.
Stacks smart contracts can
read the Bitcoin blockchain state (but not send transactions). Hashes of Stacks blocks are written into Bitcoin blocks, thus inheriting the Bitcoin latency. 
Rootstock (RSK)~\cite{lerner2022} is a sidechain which uses RBTC as its native token. RBTC is 1:1 pegged to bitcoin stored at a special address on the Bitcoin blockchain.
RSK smart contracts cannot hold native bitcoins or interact with the Bitcoin blockchain directly.
More than half of the 13 federation members need to multi-sign transactions to exchange RBTC into bitcoins. 
THORChain is a cross-chain protocol to support swaps between different blockchains via threshold ECDSA signatures, including Ethereum and Bitcoin.
Its threshold ECDSA signature protocol relies on a synchronous network assumption, making it vulnerable to node crashes and bad networking conditions. Thus, its applicability for global applications is questionable. Since the launch of its mainnet Ethereum integration, THORChain has lost millions of dollars due to successful hacks.
WBTC is a protocol for creating Bitcoin-backed wrapped tokens on various platforms, including Ethereum and Arbitrum. The protocol is rather centralized as the custodian is enacted by BitGo.
BitGo uses a multi-signature address to control the funds; however, it is not clear who controls the different keys.

The third category also includes off-chain computing with Bitcoin payments, e.g., using Bitcoin for contingent payments~\cite{banasik2016}. 
BitVM\footnote{See \url{https://bitvm.org/bitvm.pdf}.} and  FastKitten~\cite{das2019} propose to combine off-chain smart contract computing tied to Bitcoin with deposits and on-chain dispute resolution, which incentivize the correct behavior of all parties involved.
BitVM is limited to two parties, a prover and a verifier, both of which need to be online and exchange data off-chain.  The prover makes a claim that a given function evaluates for some particular inputs to some specific output together with a deposit. If a verifier submits evidence that this claim is wrong, the verifier obtains the prover's deposit.
FastKitten relies on trusted execution environments (TEEs) and thus provides confidentiality and integrity unless the TEE has been broken. FastKitten guarantees that all honest parties obtain the correct amount after execution or are reimbursed. 
It requires all parties to be known during the setup phase and they need to interact with both Bitcoin and the TEE operator repeatedly in bounded time. 
In summary, the other approaches in the third category suffer from weaker security guarantees and restrictions on the smart contracts.

\section{Conclusion}
\label{sec:conclusion}
In this paper, we have illustrated how general-purpose Bitcoin smart contracts can be executed on the Internet Computer. 
The underlying architecture is based on novel concepts such as the interconnection of the networks at the node layer and the storage of the Bitcoin blockchain state in a smart contract for quick and reliable access. 
In contrast to other approaches, this enables smart contracts to read and write to the Bitcoin state in a secure manner, fast and at low cost.
We conjecture that the Bitcoin canister is the smart contract with the largest size in existence at a state size exceeding 87 GiB.
Since the integration went live, the Bitcoin canister was queried nearly 8 million times, and it received approximately 1,900 transactions from smart contracts, which were forwarded to the Bitcoin network.
Although the integration is tailored to the capabilities of the Internet Computer, certain aspects of the architecture may prove to be useful for similar endeavors.

As the Bitcoin canister returns signed responses, verifiable by any entity that knows the public key of the corresponding subnet, it provides a \emph{trustworthy decentralized view} of the Bitcoin blockchain state. To the best of our knowledge, this is another unique property of this integration.
Thus, the presented integration provides new functionality that can be leveraged by decentralized applications as well as traditional blockchain-centric web services.

\bibliographystyle{IEEEtran}
\bibliography{bitcoin_integration}

\end{document}